\documentclass[11pt]{article}
\usepackage[margin=1.25in]{geometry}
\usepackage{ifthen}
\usepackage{graphicx}
\usepackage{epsfig}
\usepackage{amsfonts,dsfont,amssymb,amsthm,stmaryrd}
\usepackage[cmex10]{amsmath} 
\usepackage{hyperref}
\hypersetup{colorlinks=true,pdftitle="",pdftex}

\newcommand{\bc}{\begin{center}}
\newcommand{\ec}{\end{center}}
\newcommand{\bt}{\begin{tabular}}
\newcommand{\et}{\end{tabular}} 
\newcommand{\bea}{\begin{eqnarray}}
\newcommand{\eea}{\end{eqnarray}}
\newcommand{\bean}{\begin{eqnarray*}}
\newcommand{\eean}{\end{eqnarray*}}

\newcommand{\ba}{\begin{array}}
\newcommand{\ea}{\end{array}}

\def\be{\begin{eqnarray}}
\def\ee{\end{eqnarray}}
\def\ben{\begin{eqnarray*}}
\def\een{\end{eqnarray*}}

\def\elabel#1{\label{e:#1}}

\def\sq{$\Box$}

\def\qed{\ifmmode\sq\else{\unskip\nobreak\hfil
\penalty50\hskip1em\null\nobreak\hfil\sq
\parfillskip=0pt\finalhyphendemerits=0\endgraf}\fi\par\medbreak}

\newsavebox{\junk}
\savebox{\junk}[1.6mm]{\hbox{$|\!|\!|$}}

\def\til={{\widetilde =}}

\def\half{{\mathchoice{\textstyle \frac{1}{2}}%
{\frac{1}{2}}%
{\hbox{\tiny $\frac{1}{2}$}}%
{\hbox{\tiny $\frac{1}{2}$}} }}

 \def\eq#1/{(\ref{#1})}

\def\eq#1/{(\ref{e:#1})}

\newcommand{\beqn}[1]{\notes{#1}%
\begin{eqnarray} \elabel{#1}}

\newcommand{\eeqn}{\end{eqnarray} }

\newcommand{\beq}[1]{\notes{#1}%
\begin{equation}\elabel{#1}}

\newcommand{\eeq}{\end{equation}} 

\def\bdes{\begin{description}}
\def\edes{\end{description}}

\def\notes#1{}

\newcommand\independent{\protect\mathpalette{\protect\independent}{\perp}} 
\def\independent#1#2{\mathrel{\rlap{$#1#2$}\mkern2mu{#1#2}}}

\newcommand{\mR}{\mathbb{R}} 

\newcommand{\mZ}{\mathbb{Z}}

\newcommand{\mP}{\mathbb{P}}
\newcommand{\mE}{\mathbb{E}}
\newcommand{\mF}{\mathbb{F}}

\newcommand{\X}{\mathcal{X}}   
\newcommand{\Y}{\mathcal{Y}}

\newcommand{\e}{\varepsilon}

\theoremstyle{definition}
\theoremstyle{plain}
\newtheorem{thm}{Theorem}
\theoremstyle{plain}
\newtheorem{prop}{Proposition}
\theoremstyle{plain}
\newtheorem{lemma}{Lemma}
\theoremstyle{plain}
\newtheorem{corol}{Corollary}
\theoremstyle{plain}
\theoremstyle{remark}
\newtheorem{remark}{Remark}
\theoremstyle{remark}
\theoremstyle{plain}

\begin{document}

\title{Entropies of weighted sums in cyclic groups and an application to
polar codes}
\author{Emmanuel Abbe, Jiange Li, Mokshay Madiman}
\maketitle


\begin{abstract}
In this note, the following basic question is explored: in a cyclic group, 
how are the Shannon entropies of the sum and difference of i.i.d. random variables related to each other? 
For the integer group, we show that they can differ by any real number additively, but not too much multiplicatively;
on the other hand, for $\mZ/3\mZ$, the entropy of the difference is always at least as large as that of the sum.
These results are closely related to the study of more-sum-than-difference (i.e. MSTD) sets in additive combinatorics.
We also investigate polar codes for $q$-ary input channels using non-canonical kernels to construct the generator matrix,
and present applications of our results to constructing polar codes with significantly improved error probability
compared to the canonical construction.
%
\end{abstract}



\section{Introduction}

For a discrete random variable $X$ supported on a countable set $A$, its Shannon entropy $H(X)$ is defined to be
\begin{align}
H(X)=-\sum_{x\in A}\mP(X=x)\log\mP(X=x).
\end{align}
The Shannon entropy can be thought of as the logarithm of the effective 
cardinality of the support of $X$; the justification for this interpretation comes 
from the fact that when the alphabet $A$ is finite, 
$H(X)\leq \log |A|$, with equality if and only if $X$ is uniformly distributed on $A$.
This suggests an informal parallelism between entropy inequalities and set cardinality inequalities
that has been extensively explored for projections of subsets of Cartesian product sets (see, e.g.,
\cite{MT10} for a review of these and their applications to combinatorics), and more recently
for sums of subsets of a group that are of great interest in the area of additive combinatorics  \cite{TV06:book}.   
For two finite subsets $A, B$ of an additive group, the sumset $A+B$ and difference set $A-B$ are defined by
$$
A+B:=\{a+b: a\in A, b\in B\},
$$
and
$$
A-B:=\{a-b: a\in A, b\in B\}.
$$
In the trivial bound
$\max\{|A|, |B|\}\leq |A\pm B|\leq|A||B|$,
replacing the sets $A, B$ by independent discrete random variables $X, Y$ and replacing the log-cardinality 
of each set by the Shannon entropy, one obtains the entropy analogue 
\be
\max\{H(X), H(Y)\}\leq H(X\pm Y)\leq H(X)+H(Y). \label{eq: H upper-lower bound}
\ee
This is, of course, an analogy but not a proof; however, the inequality \eqref{eq: H upper-lower bound}
can be seen to be true from elementary properties of entropy.

First identified by Ruzsa \cite{Ruz09}, this connection between entropy inequalities and cardinality inequalities
in additive combinatorics has been studied extensively in the last few years. Useful tools in additive 
combinatorics have been developed in the entropy setting, such as Pl\"{u}nnecke-Ruzsa inequalities
by Madiman, Marcus and Tetali \cite{MMT12}, 
and Freiman-Ruzsa and Balog-Szemer\'edi-Gowers theorems by Tao \cite{Tao10}.
Much more work has also recently emerged on related topics, such as 
efforts towards an entropy version of the Cauchy-Davenport inequality \cite{HAT14, JA14, WWM14:isit, WM15:isit},
an entropy analogue of the doubling-difference inequality \cite{MK10:isit},
extensions from discrete groups to locally compact abelian groups \cite{KM14, MK15},
and applications of additive combinatorics in information theory \cite{LP08:ieeei, Mad08:itw, CZ08, EO09, WSV12}.


In an abelian group, since addition is commutative while subtraction is not, two generic elements generate one sum but two differences. 
Likely motivated by this observation, 
the following conjecture is contained in H. T. Croft's Research Problems, 1967:
\begin{quote}
``Let $A =\{a_1, a_2,\ldots, a_N\}$ be a finite set of integers, and define
$A+A = \{a_i+a_j:  1 \leq i, j \leq N\}$ and $A-A = \{a_i-a_j:  1 \leq i, j \leq N\}$.
Prove that $A-A$ always has more members than $A+A$, unless $A$ is symmetric about 0.''
\end{quote}
However, that is not always the case. In 1969, I. J. Marica \cite{Mar69} showed that the conjecture 
(which he attributed to J. H. Conway) is false by exhibiting the set 
$A = \{1, 2, 3, 5, 8, 9, 13, 15, 16\}$, for which $A+A$ has 30 elements and $A-A$ has 29 elements.
According to Nathanson \cite{Nat07:2}, Conway himself had already found the MSTD set $\{0, 2, 3, 4, 7, 11, 12, 14\}$ in the 1960's,
thus disproving the conjecture attributed to him.
Subsequently, S. Stein \cite{Ste73} showed that one can construct sets $A$ for which the ratio $|A-A|/|A+A|$ is
as close to 0 or as large as we please; apart from his own proof, he observed that such constructions also follow by adapting
arguments in an earlier work of S. Piccard \cite{Pic42} that focused on the Lebesgue measure
of $A+A$ and $A-A$ for subsets $A$ of $\mathbb{R}$. 
A stream of recent papers aims to quantify how rare or frequent MSTD sets are
(see, e.g., \cite{MO07, HM09} for work on the integers, and \cite{Zha10:2} for finite abelian groups
more generally), or try to provide denser constructions of infinite families of MSTD sets (see, e.g., \cite{MOS10, Zha10:3});
however these are not directions we will explore in this note.


Since convolutions of uniforms are always distributed on the sumset
of the supports, but are typically not uniform distributions, it is not immediately obvious from 
the Conway and Marica constructions whether there exist i.i.d. random variables
$X$ and $Y$ such that $H(X+Y)>H(X-Y)$. The purpose of this note is to explore this and related questions. For example, one natural related
question to ask is for some description of the coefficient $\lambda\in \{1, \ldots, |G|\}$ that maximizes
$H(X+\lambda Y)$ for $X, Y$ drawn i.i.d. from some distribution in $G$; restricting the choice of 
coefficients to $\{+1, -1\}$ would correspond to the sum-difference question. 
This question is motivated by applications
to the class of polar codes, which is a very promising class of codes that has attracted much recent attention
in information and coding theory.
Specifically, we show that over $\mF_q$, the ``spread'' of the polar martingale can be significantly 
enlarged by using 
optimized kernels rather than the original kernel $\bigl[\begin{smallmatrix} 
      1 & 0 \\
      1 & 1 \\
   \end{smallmatrix}\bigr]$. In some cases, this leads to significant improvements on the error probability of polar codes, 
even at low block lengths like 1024. We also consider additive noise channels, and show that the improvement 
is particularly significant when the noise distribution is concentrated on ``small'' support. 

This note is organized as follows. In Section~\ref{sec:examples}, we show that entropies of sums (of i.i.d. random variables) are never greater than
entropies of differences for random variables taking values in the cyclic group $\mZ/3\mZ$; however this fails for larger groups, and in particular we show that there always exist distributions on finite cyclic groups of order at least 21
such that $H(X+Y)>H(X-Y)$. In Section \ref{sec:additive} and Section \ref{sec:multiplicative}, we explore more quantitative questions-- that is, we ask not only
what the ordering of $H(X+Y)$ and $H(X-Y)$ may be, but how different these can be in either direction; the finding here is that
on $\mZ$, these can differ by arbitrarily large amounts additively, but not too much multiplicatively. 
Finally, in Section~\ref{sec: polar}, we explore the question about
entropies of weighted sums mentioned at the end of the previous paragraph,
and describe the applications to polar codes as well.

\section{Comparing entropies of sums and differences} 


\subsection{Basic examples} \label{sec:examples}

We start by considering the smallest group in which the sum and difference are distinct, namely $\mZ/3\mZ$.
Let $p=(p_0, p_1, p_2)$ be a probability distribution on $\mZ/3\mZ$, and let $H(p)$ be its Shannon entropy. We denote by $\|p-U\|_2$ the Euclidean distance between $p$ and the uniform distribution $U=(\frac{1}{3}, \frac{1}{3}, \frac{1}{3})$. For any fixed $0\leq t\leq\log3$, the following lemma verifies the ``triangular'' shape of the entropy circle $H(p)=t$. 

\begin{lemma}\label{lem:triangular}
Let $p$ be a probability distribution on the entropy circle $H(p)=t$ such that $p_0\geq p_1\geq p_2$. Then the distance $\|p-U\|_2$ is an increasing function of $p_0$.
\end{lemma}

\begin{proof}
If $t=0$, then $p$ has to be the deterministic distribution $(1, 0, 0)$. In this case, we have $\|p-U\|_2=\sqrt{2/3}$. If $t=\log 3$, we have $p=U$ and $\|p-U\|_2=0$. In the following, we may assume that $0<t<\log 3$. The condition $p_0+p_1+p_2=1$ yields
\begin{align}\label{eq:diff-ident1}
1+\frac{d p_1}{d p_0}+\frac{d p_2}{d p_0}=0.
\end{align}
The entropy identity $H(p)=t$ implies
\begin{align}\label{eq:diff-ident2}
(\log p_0+1)+(\log p_1+1)\frac{d p_1}{d p_0}+(\log p_2+1)\frac{d p_2}{d p_0}=0
\end{align}
The above two identities give us that
\begin{align}\label{eq:dp1}
\frac{d p_1}{d p_0}=\frac{\log p_0-\log p_2}{\log p_2-\log p_1}
\end{align}
and
\begin{align}\label{eq:dp2}
\frac{d p_2}{d p_0}=\frac{\log p_0-\log p_1}{\log p_1-\log p_2}.
\end{align}
Using identities \eqref{eq:diff-ident1}, \eqref{eq:dp1} and \eqref{eq:dp2}, we have
\begin{eqnarray*}
\frac{1}{2}\cdot\frac{d}{dp_0}\|p-U\|_2 &=&\sum_{i=0}^2\left(p_i-\frac{1}{3}\right)\frac{d p_i}{d p_0}\\
&=& p_0+p_1\frac{\log p_0-\log p_2}{\log p_2-\log p_1}+p_2\frac{\log p_0-\log p_1}{\log p_1-\log p_2}\\
&=&(p_0-p_1)\frac{\log p_0-\log p_2}{\log p_1-\log p_2}-(p_0-p_2)\frac{\log p_0-\log p_1}{\log p_1-\log p_2}\\
&=&\frac{(p_0-p_1)(p_0-p_2)}{\log p_1-\log p_2}\left(\frac{\log p_0-\log p_2}{p_0-p_2}-\frac{\log p_0-\log p_1}{p_0-p_1}\right)\\
&\geq&0
\end{eqnarray*}
The last inequality follows from the assumption that $p_0\geq p_1\geq p_2$ and the concavity of the logarithmic function.
\end{proof}

Now we can show that the entropy of the sum of two i.i.d. random variables taking values in $\mZ/3\mZ$ can never exceed the entropy of their difference.
We use basic facts about the Fourier transform on finite groups, which can be found, e.g., in \cite{SS03:1:book}.

\begin{thm}\label{thm:Z3}
Let $X, Y$ be i.i.d.\ random variables taking values in $\mZ/3\mZ$, then we have
\begin{align}
H(X+Y) \leq H(X-Y).
\end{align}
\end{thm}

\begin{proof}
Let $p=(p_0, p_1, p_2)$ be the distribution of $X$. Since $Y$ is an independent copy of $X$, we can see that $-Y$ has distribution $q=(p_0, p_2, p_1)$. Then the distributions of $X+Y$ and $X-Y$ can be written as $p\star p$ and $p\star q$, respectively, where `$\star$' is the convolution operation. Let $\widehat{p}=(\widehat{p}_0, \widehat{p}_1, \widehat{p}_2)$ be the Fourier transform of $p$ with Fourier coefficients defined by
$$
\widehat{p}_j=\sum_{k=0}^2p_ke^{-i2\pi jk/3},~~j=0, 1, 2.
$$
One basic property of the Fourier transform asserts that
\begin{align}\label{eq:conjugate}
\widehat{q}_j=\overline{\widehat{p}_j},
\end{align}
where $\overline{\widehat{p}_j}$ is is the conjugate of $\widehat{p}_j$. We also have
\begin{align}\label{eq:f-con}
(\widehat{p\star q})_j=\widehat{p}_j\cdot\widehat{q}_j,
\end{align}
which holds for general distributions $q$. The Parseval-Plancherel identity says
\begin{align}\label{eq:pp}
\|\widehat{p}\|_2^2=3\|p\|_2^2.
\end{align}
Using the identities \eqref{eq:conjugate}, \eqref{eq:f-con} and \eqref{eq:pp}, we have
$$
\|p\star p\|_2 = \|p\star q\|_2,
$$
which implies
$$
\|p\star p-U\|_2=\|p\star q-U\|_2.
$$ 
It is not hard to see that $X-Y$ is symmetric with $(p\star q)_0\geq(p\star q)_1=(p\star q)_2$.
Using Lemma \ref{lem:triangular}, we can see that the entropy circle passing through $p\star q$ lies inside the Euclidean circle centered at $U$ with radius $\|p\star q-U\|_2$. Thus the distribution $p\star p$ is on an entropy circle with entropy not greater than $H(p\star q)$. Then we have the desired statement.
\end{proof}



The property in Theorem \ref{thm:Z3} fails to hold for larger cyclic groups; we demonstrate this 
by discussing three specific examples of i.i.d. random variables $X, Y$ such that the 
entropy of their sum is larger than the entropy of their difference. 

\begin{enumerate}

\item 
For Conway's MSTD set $A=\{0, 2, 3, 4, 7, 11, 12, 14\}$, we have $|A+A|=26$ and $|A-A|=25$. 
Let $X, Y$ be independent random variables uniformly distributed on $A$. Straightforward calculations show that
$$
H(X+Y)-H(X-Y)=\frac{1}{64}\log\frac{282429536481}{215886856192}>0.
$$


\item 
The second example is based on the set $A=\{0, 1, 3, 4, 5, 6, 7, 10\}$ with $|A+A|=|A-A|=19$. Let $X, Y$ be independent random variables uniformly distributed on $A$. Then we have
$$
H(X+Y)-H(X-Y)=\frac{1}{64}\log\frac{5^{10}\cdot8^{10}}{3^6\cdot 7^7}>0.
$$


\item 
The group $\mZ/12\mZ$ is the smallest cyclic group that contains a MSTD set. Let $A=\{0, 1, 2, 4, 5, 9\}$. It is easy to check that $A$ is a MSTD set since $A+A=\mZ/12\mZ$ and $A-A=(\mZ/12\mZ)\backslash\{6\}$. We let $X, Y$ be independent random variables uniformly distributed on $A$. Then we have
$$
H(X+Y)-H(X-Y)=\frac{1}{36}\log\frac{3^{34}}{20^{10}}>0.
$$
\end{enumerate}

\begin{remark}
Applying linear transformations, we can get infinitely many MSTD sets of $\mZ$ from Conway's MSTD set. 
Correspondingly, one can get as many ``MSTD'' random variables as one pleases. 
Thus MSTD sets are useful in the construction of ``MSTD'' random variables;
however we can also construct ``MSTD'' random variables supported on non-MSTD sets as shown by the second example.
\end{remark}

\begin{remark}
Hegarty \cite{Heg07} proved that there is no MSTD set in $\mZ$
of size 7 and, up to linear transformations, Conway's set is the unique MSTD set in $\mZ$ of size 8. 
We do not know the smallest support of ``MSTD'' random variables taking values in $\mZ$,
although 8 is clearly an upper bound.
\end{remark}

\begin{remark}
We also do not know the smallest $m$ such that there exist ``MSTD'' random variables 
taking values in $\mZ/m\mZ$; however, the third example shows that this $m$ cannot be greater than $12$.
\end{remark}

\subsection{Achievable differences}
\label{sec:additive}

We first briefly introduce the construction of Stein \cite{Ste73} of finite subsets $A_k\subset\mZ$ 
such that the ratio $|A_k-A_k|/|A_k+A_k|$ can be arbitrarily large or small when $k$ is large. 
Using this construction we will give an alternate proof of the result of Lapidoth and Pete \cite{LP08:ieeei}, 
which asserts that $H(X-Y)$ can exceed $H(X+Y)$ by an arbitrarily large amount. 

Let $A, B\subset\mZ$ be two finite subsets. Suppose that the gap between any two consecutive elements of $B$ is sufficiently large. For any $b\in B$, the set $b+A$ represents a relatively small fluctuation around $b$. Large gaps between elements of $B$ will imply that $(b+A)\cap(b'+A)=\emptyset$ for distinct $b, b'\in B$. Then we will have $|A+B|=|A||B|$. For $m\in\mZ$ large, this argument implies that $|A+m\cdot A|=|A|^2$, where $m\cdot A:=\{ma: a\in A\}$. Therefore, the following equations hold simultaneously for sufficiently large $m_0\in\mZ$ 
(which depends on $A, A-A$ and $A+A$):
$$
|A+m_0\cdot A|=|A|^2,
$$
$$
|(A+m_0\cdot A)-(A+m_0\cdot A)| = |(A-A)+m_0\cdot(A-A)|
= |A-A|^2,
$$
and
$$
|(A+m_0\cdot A)+(A+m_0\cdot A)|=|A+A|^2.
$$
Repeating this argument, we can get a sequence of sets $A_k$, defined by
\begin{align}
A_k=A_{k-1}+m_{k-1}A_{k-1}, \label{eq: iteration}
\end{align}
where $A_0=A$, $m_{k-1}\in\mZ$ sufficiently large, with the following properties
\begin{align} \label{eq: A_k property}
|A_k|=|A|^{2k},~~
|A_k\pm A_k|=|A\pm A|^{2k}.
\end{align}

Now we are ready to reprove the result of Lapidoth and Pete \cite{LP08:ieeei}.

\begin{thm}\cite{LP08:ieeei}\label{thm:arb-lp}
For any $M>0$, there exist i.i.d. $\mZ$-valued random variables $X, Y$ with finite entropy such that
$$
H(X-Y)- H(X+Y) > M.
$$
\end{thm}

\begin{proof}
Recall the following basic property of Shannon entropy
\begin{align}
0\leq H(X)\leq\log|\text{range of}~X|. \label{eq: H upper bound}
\end{align}
We let $X_k, Y_k$ be independent random variables uniformly distributed on the set $A_k$ obtained by the iteration equation \eqref{eq: iteration}. Using the right hand side of \eqref{eq: H upper bound} and the properties given by \eqref{eq: A_k property}, we have
\begin{align}
H(X_k+Y_k) \leq\log|A_k+A_k|=2k\log|A+A|. \label{eq: X_k+Y_k}
\end{align}
Since $X_k, Y_k$ are independent and uniform on $A_k$, for all $x\in A_k-A_k$, we have
$$
\mP(X_k-Y_k=x)\geq|A_k|^{-2}.
$$
Notice the fact that $-t\log t$ is increasing over $(0, 1/e)$. When $k$ is large enough, we have
\begin{eqnarray}
H(X_k-Y_k) &\geq& \frac{|A_k-A_k|}{|A_k|^2}\log |A_k|^2\nonumber\\
&=&4k\log|A|\left(\frac{|A-A|}{|A|^2}\right)^{2k}. \label{eq: X_k-Y_k}
\end{eqnarray}
For any $k\in\mZ^+$, we can always find a set $A\subset\mZ$ with $k^2$ elements such that the set $A-A$ achieves the possible maximal cardinality,
\begin{align}
|A|=k^2, ~|A-A|=|A|^2-|A|+1. \label{eq: choice of A}
\end{align}
Combining \eqref{eq: X_k+Y_k}, \eqref{eq: choice of A} and the trivial bound 
$$
|A+A|\leq\frac{|A|(|A|+1)}{2},
$$
we have that for $k$ large
\begin{eqnarray*}
H(X_k+Y_k) &\leq& 2k\log\frac{|A|(|A|+1)}{2}\nonumber\\
&=& 8k\log k-2k\log2+2k\log(1+k^{-2})\nonumber\\
&=&  8k\log k-2k\log2+o(1).
\end{eqnarray*}
Combining \eqref{eq: X_k-Y_k} and \eqref{eq: choice of A}, we have
\begin{eqnarray*}
H(X_k-Y_k) &\geq&8k\log k\left(1-k^{-2}+k^{-4}\right)^{2k}\nonumber\\
&=& 8k\log k\exp(2k(-k^{-2}+O(k^{-4})))\nonumber\\
&=& 8k\log k(1-2k^{-1}+O(k^{-2}))\nonumber\\
&=& 8k\log k-16\log k+o(1).
\end{eqnarray*}
Therefore we have
$$
H(X_k-Y_k)-H(X_k+Y_k)=2k\log2-16\log k+o(1).
$$
Then the statement follows from that $k$ can be arbitrarily large.
\end{proof}

We observe that the following complementary result is also true.

\begin{thm}\label{thm:arb-mstd}
For any $M>0$, there exist i.i.d. $\mZ$-valued random variables $X, Y$ with finite entropy such that
$$
H(X+Y)-H(X-Y)> M.
$$
\end{thm}

\begin{remark}
The previous argument can not be used to prove this result. If we proceed the same argument, we will see that the lower bound of $H(X_k+Y_k)$ similar to \eqref{eq: X_k-Y_k} will be really bad. The reason is that
$$
\left(\frac{|A+A|}{|A|^2}\right)^{2k}\rightarrow 0
$$ 
exponentially fast. Both Theorems~\ref{thm:arb-lp} and \ref{thm:arb-mstd} can be proved using a probabilistic 
construction of Ruzsa~\cite{Ruz92:1} on the existence of large
additive sets $A$ with $|A-A|$ very close to the maximal value $|A|^2$, but $|A+A|\leq n^{2-c}$ for some explicit absolute constant $c>0$; and similarly with the roles of $A-A$ and $A+A$ reversed.
\end{remark}

In fact, we have the following stronger result.

\begin{thm}\label{thm:arb}
For any $M\in \mR$, there exist i.i.d. $\mZ$-valued random variables $X, Y$ with finite entropy such that 
$$
H(X+Y)-H(X-Y)=M.
$$
\end{thm}

\begin{proof}
Let $X$ be a random variable taking values in $\{0, 1, \cdots, n-1\}\subset\mZ$. Then $H(X+Y)-H(X-Y)$ is a continuous function of 
the probability mass function of $X$, which consists of $n$ real variables. We can assume that $n$ is large enough if necessary. From the discussion in Section \ref{sec:examples}, we know that this function can take both positive and negative values. (For instance Theorem \ref{thm:Z3} implies that a binary distribution can give us negative difference, and the uniform distribution on Conway's MSTD set will yield positive difference). Since the function is continuous, the intermediate value theorem implies that its range must contain an open interval $(a, b)$ with $a<0<b$. Let $X_1, \cdots X_k$ be $k$ independent copies of $X$ and we define $X'=(X_1, \cdots, X_k)$. Let $Y'$ be an independent copy of $X'$. Then we have
$$
H(X'+Y')-H(X'-Y')=k [H(X+Y)-H(X-Y)] .
$$
The range of $H(X'+Y')-H(X'-Y')$ will contain $(ka, kb)$. This difference can take any real number since $k$ can be arbitrarily large. The random variables $X', Y'$ take finite values of $\mZ^k$. Using the linear transformation $(x_1, \cdots, x_k)\to x_1+dx_2+\cdots+d^{k-1}x_k$, we can map $X, Y$ to $\mZ$-valued random variables. This map preserves entropy as $d$ is large enough. So these $\mZ$-valued random variables will have the desired property. 
\end{proof}

Recall that, for a real-valued random variable $X$ with the density function $f(x)$, the differential entropy $h(X)$ is defined by
\begin{align}
h(X)=-\mE\log f(X) = -\int_{\mR} f(x)\log f(x)dx .
\end{align}

\begin{thm}
For any $M\in \mR$, there exist i.i.d. real-valued random variables $X, Y$ with finite differential entropy such that
\begin{align}
h(X+Y)-h(X-Y)=M.
\end{align}
\end{thm}

\begin{proof}
From Theorem~\ref{thm:arb}, we know that there exist $\mZ$-valued random variables $X', Y'$ with the desired property. Let $U, V$ be independent random variables uniformly distributed on $(-1/4, 1/4)$, which are also independent of $(X', Y')$. Then we define
$X=X'+U$ and $Y=Y'+V$.
Elementary calculations will show that
$$
h(X+Y)=H(X'+Y')+h(U+V),
$$
and
$$
h(X-Y)=H(X'-Y')+h(U-V).
$$
Since $U, V$ are symmetric, $U+V$ and $U-V$ have the same distribution. Therefore, we have
$$
h(X+Y)-h(X-Y)=H(X'+Y')-H(X'-Y').
$$
Then the theorem follows.
\end{proof}


\begin{remark}
In the set cardinality setting, Nathanson \cite{Nat07:1} raised the question: what are the possible values of $|A+A|-|A-A|$ for finite subsets $A\subset \mZ$? Martin and O'Bryant \cite{MO07} proved that for any $k\in \mZ$ there exists $A$ such that $|A+A|-|A-A|=k$; 
this was also independently obtained by Hegarty \cite{Heg07}.
\end{remark}

\subsection{Entropy analogue of Freiman-Pigarev inequality}
\label{sec:multiplicative}

We proved that the entropies of the sum and difference of two i.i.d. random variables can differ by an arbitrary amount additively. 
However we will show that they do not differ too much multiplicatively.


In additive combinatorics, for a finite additive set $A$, the doubling constant $\sigma[A]$ is defined as
\begin{align}
\sigma[A]=\frac{|A+A|}{|A|}.
\end{align}
Similarly the difference constant $\delta[A]$ is defined by
\begin{align}
\delta[A]=\frac{|A-A|}{|A|}.
\end{align}
It was first observed by Ruzsa \cite{Ruz78} that
\begin{align}
\delta[A]^{1/2}\leq\sigma[A]\leq\delta[A]^3. \label{eq: doubling-difference}
\end{align}
The upper bound can be improved down to $\delta[A]^2$ using Pl\"{u}nnecke inequalities. Thus a finite additive set has small doubling constant if and only if its difference constant is also small. In the entropy setting, we have
\begin{align}
\frac{1}{2}\leq\frac{H(X+Y)-H(X)}{H(X-Y)-H(X)}\leq 2 \label{eq: H doubling-difference}
\end{align}
for i.i.d. random variables $X, Y$.
The upper bound was proved by Madiman \cite{Mad08:itw} and the lower bound was proved independently by Ruzsa \cite{Ruz09} and Tao \cite{Tao10}. The inequalities also hold for differential entropy \cite{MK10:isit, KM14}, and in fact for entropy with respect to the
Haar measure on any locally compact abelian group \cite{MK15}. In other words, after subtraction of $H(X)$, the entropies of the sum and the difference of two i.i.d. random variables are not too different. We observe that the entropy version \eqref{eq: H doubling-difference}
of the doubling-difference inequality implies the entropy analogue of the following result proved by Freiman and Pigarev \cite{FP73}:
\begin{align}
|A-A|^{3/4}\leq |A+A|\leq|A-A|^{4/3}.
\end{align}

\begin{thm}
Let $X, Y$ be i.i.d. discrete random variables with finite entropy, then we have
\begin{align}
\frac{3}{4}<\frac{H(X+Y)}{H(X-Y)}<\frac{4}{3}.
\end{align}
\end{thm}

\begin{proof}
The basic fact of Shannon entropy \eqref{eq: H upper-lower bound} implies that $H(X+Y)=0$ if and only if $H(X-Y)=0$. In this case, the above theorem is true if we define $0/0=1$. So we assume that neither $H(X+Y)$ nor $H(X-Y)$ is 0. For the upper bound, we have
\begin{eqnarray*}
\frac{H(X+Y)}{H(X-Y)} &=& \frac{H(X+Y)}{H(X-Y)-H(X)+H(X)}\\
&\leq& \frac{H(X+Y)}{(H(X+Y)-H(X))/2+H(X)}\\
&=&\frac{2H(X+Y)}{H(X+Y)+H(X)}\\
&<&\frac{4}{3}
\end{eqnarray*}
The second step follows from the upper bound in \eqref{eq: H doubling-difference} and the fact that Shannon entropy is non-negative. The last step uses the right hand side of \eqref{eq: H upper-lower bound} and the fact that, in the i.i.d. case, $``="$ of the upper bound happens only when $X$ takes on a single value, i.e. $H(X)=0$. 
The lower bound can be proved in a similar way.
\end{proof}

 \begin{remark}
It is unknown if the inequality \eqref{eq: H doubling-difference} is best possible. Suppose that, for some $\alpha\in(1, 2)$, we have
\ben
\alpha^{-1}\leq\frac{H(X+Y)-H(X)}{H(X-Y)-H(X)}\leq \alpha.
\een
Using the same argument, the above theorem can be improved to
\ben
\frac{\alpha+1}{2\alpha}<\frac{H(X+Y)}{H(X-Y)}<\frac{2\alpha}{\alpha+1}.
\een
\end{remark}

\begin{remark}

The above theorem does not hold for continuous random variables. For example, let $X$ be an exponential random variable with parameter $\lambda$, and $Y$ be an independent copy of $X$. Then $X+Y$ satisfies the Gamma distribution $\Gamma(2, \lambda^{-1})$ with the differential entropy 
$$
h(X+Y)=1+\gamma-\log\lambda\approx1.577-\log\lambda,
$$
where $\gamma$ is the Euler's constant. On the other hand, $X-Y$ has the Laplace distribution $\text{Laplace}(0, \lambda^{-1})$ with the differential entropy
$$
h(X-Y)=1+\log2-\log\lambda\approx1.693-\log\lambda.
$$
We can see that
$$
\lim_{\lambda\rightarrow (2e)^+}\frac{h(X+Y)}{h(X-Y)}=\infty,
$$
and
$$
\lim_{\lambda\rightarrow (2e)^-}\frac{h(X+Y)}{h(X-Y)}=-\infty.
$$
\end{remark}


\section{Weighted sums and polar codes} \label{sec: polar}

\subsection{Polar codes: introduction}

Polar codes, invented by Ar{\i}kan \cite{Ari09} in 2009, achieve the capacity of arbitrary binary-input symmetric discrete memoryless channels. 
Moreover, they have low encoding and decoding complexity and an explicit construction. Consequently they have attracted a great
deal of attention in recent years. In order to discuss polar codes more precisely, we now recall some standard terminology
from information and coding theory.

As is standard practice in information theory, 
we use $U^k$ to denote $(U_1, \ldots, U_k)$,
and $I(X;Y|Z)$ to denote the conditional mutual information between $X$ and $Y$ given $Z$,
which is defined by
\ben
I(X;Y|Z)= H(X,Z)+H(Y,Z)-H(X,Y,Z)- H(Z) .
\een
It is well known, and also trivial to see, that the conditional entropy $H(X|Y)$, defined as the mean using the distribution of $Y$
of $H(X|Y=y)$, satisfies the ``chain rule'' $H(Y)+H(X|Y)=H(X,Y)$, so that $I(X;Y|Z)=H(X|Z)-H(X|Y,Z)$.
The mutual information between $X$ and $Y$, namely $I(X;Y)=H(X)-H(X|Y)$,
emerges in the case where there is no conditioning.
In particular, $I(X;Y|Z)=0$ if and only if $X$ and $Y$ are conditionally independent given $Z$.
Furthermore, one also has the chain rule for mutual information, which states that
$I(X; Y,Z)= I(X;Z) + I(X;Y|Z)$.

A major goal in coding theory is to obtain efficient codes that achieve the Shannon capacity on a discrete memoryless channel. A memoryless channel is defined first by a ``one-shot'' channel $W$, which is a stochastic kernel from an input alphabet $\X$ to an output alphabet $\Y$ (i.e., for each $x \in \X$, $W(\cdot |x)$ is a probability distribution on $\Y$), and the memoryless extension of $W$ for length $n$ vectors is defined by 
\begin{align}
W^{(n)}(y^n|x^n) = \prod_{i=1}^n W(y_i|x_i), \quad x^n \in \X^n, y^n \in \Y^n. 
\end{align}
To simplify the notation, one often makes a slight abuse of notation, writing $W^{(n)}$ as $W$.

A linear code of block length $n$ on an alphabet $\X=\mathbb{F}$ (which must be a field) is a subspace of $\mathbb{F}^n$. The vectors in the subspace are often called the codewords. A linear code is equivalently defined by a generator matrix, i.e., a matrix with entries in the field whose rows form a basis for the code. 
If the dimension of the code is $k$, and if $G$ is a $k \times n$ generator matrix for the linear code, the codewords are given by the span of the rows of $G$, i.e., all multiplications $uG$ where $u$ is a $1 \times k$ vector over the field. We refer to \cite{CT06:book,Rot06:book} for a more detailed introduction to information and coding theory.

In polar codes, the generator matrix of block length $n$ 
is obtained by deleting\footnote{If the channel is symmetric the generator matrix is indeed obtained by deleting rows, 
otherwise in addition to deleting rows one may also have to translate the codewords (affine code).} 
some rows of the matrix 
$G_n= \bigl[\begin{smallmatrix} 
      1 & 0 \\
      1 & 1 \\
   \end{smallmatrix}\bigr]^{\otimes \log_2 n}$. 
Which rows to delete depends on the channel and the targeted error probability (or rate). 
For a symmetric discrete memoryless channel $W$, the rows to be deleted are indexed by 
 \begin{align}
 \mathcal{B}_{\e,n}:=\{i \in [n] : I(U_i;Y^n U^{i-1}) \leq 1- \e\},
 \end{align} 
where $\e$ is a parameter governing the error probability, 
the vector $U^n$ has i.i.d.\ components which are uniform on the input alphabet, 
$X^n=U^n G_n$, and $Y^n$ is the output of $n$ independent uses of $W$ when $X^n$ is the input. 

To see the purpose of the transform $G_n$, consider the case $n=2$ first. 
Applying $G_2$ to the vector $(U_1,U_2)$ yields  
\begin{align} \nonumber
X_1&=U_1+U_2,\\ 
X_2&=U_2. \nonumber
\end{align}
Transmitting $X_1$ and $X_2$ on two independent uses of a binary input channel $W$ leads to two output variables $Y_1$ and $Y_2$; recall that this means that $Y_1$ (or $Y_2$) is a random variable whose distribution is given by $W(\cdot | x)$ where $x$ is the realization of $X_1$ (or $X_2$).  If we look at the mutual information between the vectors $X^2=(X_1,X_2)$ and $Y^2=(Y_1,Y_2)$, since the pair of components $(X_1,Y_1)$ and $(X_2,Y_2)$ are mutually independent, the chain rule yields 
\begin{align}
I(X^2;Y^2)=I(X_1;Y_1) + I(X_2;Y_2) = 2 I(W), \label{eq1}
\end{align}
where $I(W)$ is defined as the mutual information of the one-shot channel $W$ with a uniformly distributed  input. Further, since the transformation $G_2$ is one-to-one, and since the mutual information is clearly invariant under one-to-one transformations of its arguments (the mutual information depends only on the joint distribution of its arguments), we have that 
\begin{align}
I(U^2;Y^2)=I(X^2;Y^2). \label{eq2}
\end{align}
If we now apply the chain rule to the left hand side of previous equality, the dependencies in the components of $U^2$ obtained by mixing $X^2$ with $G_2$ lead this time to two different terms, namely,  
\begin{align}
I(U^2;Y^2)=I(U_1;Y^2) + I(U_2;Y^2,U_1). \label{eq3}
\end{align}
Putting back \eqref{eq1}, \eqref{eq2}, and \eqref{eq3} together, we have that
\begin{align}
I(W)= \frac{1}{2} \left( I(U_1;Y^2) + I(U_2;Y^2,U_1) \right). \label{eq4}
\end{align}
Now, the above is interesting because the two terms in the right-hand side are precisely not equal. 
In fact,  $I(U_2;Y^2,U_1)$ must be greater than its counter-part without the mixing of $G_2$, i.e.,  $I(U_2;Y^2,U_1) \geq I(X_2;Y_2)=I(W)$. To see this, note that 
\begin{align*}
I(U_2;Y^2,U_1) &= H(U_2) - H(U_2|Y^2,U_1) \\
&\geq H(U_2) - H(U_2|Y^2)\\
&= H(X_2) - H(X_2|Y_2)\\
&= I(X_2;Y_2)
\end{align*}
where the inequality above uses the fact that conditioning can only reduce entropy, hence dropping the variable $U_1$ in $H(U_2|Y^2,U_1)$ can only increase the entropy.
Further, one can check that besides for degenerated cases where $W$ is deterministic or fully noisy (i.e., making input and output independent), $I(U_2;Y^2,U_1)$ is strictly larger than $I(X_2;Y_2)$.
Thus, the two terms in the right-hand side of \eqref{eq4} are respectively lesser and greater that $I(W)$, but they average out to the original amount $I(W)$.

In summary, out of two independent copies of the channel $W$, the transform $G_2$ allows us to create two new synthetic channels 
\begin{align*}
W^-&: U_1 \to  Y_1, Y_2 \\
W^+&: U_2 \to  Y_1, Y_2, U_1 
\end{align*}
that have respectively a worse and better mutual information
\begin{align*}
&I(W^-) \leq I(W) \leq I(W^+).
\end{align*}
while overall preserving the total amount of mutual information
\begin{align*}
&I(W)=\frac{1}{2} (I(W^+) + I(W^-)).
\end{align*}
The key use of the above phenomena, is that if one wants to transmit only one bit (uniformly drawn), using $W^+$ rather than $W$ leads to a lower error probability since the channel $W^+$ carries more information. One can then iterate this argument several times and hope obtaining a subset of channels of very high mutual information, on which bits can be reliably transmitted. After $\log_2n$ iterations, one obtains the synthesized channels $U_i \mapsto (Y^n, U^{i-1})$. 
Thus, for a given number of information bits to be transmitted (i.e., for a given rate), one can select the channels with the largest mutual informations to minimize the error probability. As explained in the next section, the phenomenon of {\it polarization} happens in the sense that as $n$ tends to infinity, the synthesized channels have mutual information tending to either 0 or 1 (besides for a vanishing fraction of exceptions). Hence, sending information bits through the high mutual information channels (equivalently, deleting rows of $G_n$ corresponding to low mutual information channels) allows one to achieve communication rates as large as the mutual information of the original binary input channel.  The construction extends to $q$-ary input alphabets when $q$ is prime\footnote{If $q$ is a power of a prime, 
and one uses modulo $q$ operations, the polarization still occurs but to multiple levels, as shown independently
by Park and Barg \cite{PB13} and Sahebi and Pradhan \cite{SP13}.}, using the same matrix $G_n= \bigl[\begin{smallmatrix} 
      1 & 0 \\
      1 & 1 \\
   \end{smallmatrix}\bigr]^{\otimes \log_2 n}$, while carrying the operations over $\mF_q$.

It is tempting to investigate what happens if one keeps the Kronecker structure of the generator matrix but modifies the kernel 
$\bigl[\begin{smallmatrix} 
      1 & 0 \\
      1 & 1 \\
   \end{smallmatrix}\bigr]$.
For binary input alphabets, there is no other interesting choice (up to equivalent permutations).  
 In Mori and Tanaka \cite{MT10:itw}, the error probability of non-binary polar codes constructed on the basis of Reed-Solomon matrices is calculated using numerical simulations on $q$-ary erasure channels. It is confirmed that 4-ary polar codes can have significantly better performance than binary polar codes. Our goal here is to investigate potential improvements at {\it finite block length} using modified kernels over $\mF_q$. We propose to pick kernels not by optimizing the polar code exponent as in \cite{MT10:itw} but by maximizing the polar martingale spread. This connects to the object of study in this paper, as explained next. The resulting improvements are illustrated with numerical simulations. 
 
\subsection{Polar martingale}
In order to see that polarization happens, namely that
\begin{align}
\frac{1}{n} |\{i \in [n] : I(U_i; Y^n ,U^{i-1}) \in (\e,1-\e) \}| \to 0,
\end{align}
it is helpful to rely on a random process having a uniform measure on the possible realizations of $I(U_i; Y^n U^{i-1})$. Then, counting the number of such mutual informations in $(\e,1-\e)$ can be obtained by evaluating the probability that the process lies in this interval. 
The process is defined by taking $\{B_n\}_{n \geq 1}$ to be i.i.d. random variables uniform  on $\{-,+\}$ and the binary (or $q$-ary with $q$ prime) random input channels $\{W_n, \ n \geq 0\}$ are defined by 
\begin{align}
& W_0 := W, \notag \\
& W_n := W_{n-1}^{B_n}, \quad \forall n \geq 1. \label{ind} 
\end{align}
Then the polarization result can be expressed as
\begin{align}
\mP\{ I(W_n) \in (\e,1-\e) \}  \to 0.
\end{align}
The process $I(W_n)$ is particularly handy as it is a bounded martingale with respect to the filtration $B_n$. This is a consequence of the balance equation derived in \eqref{eq4}. Therefore, $I(W_n)$ converges almost surely, which means that almost surely, for any $\e>0$ and $n$ large enough, $|I(W_{n+1})- I(W_n)| = I(W_n^+) - I(W_n)<\e$. Since for $q$-ary input channels ($q$ prime), the only channels for which $I(W^+)-I(W)$ is arbitrarily small is when $I(W)$ is arbitrarily close to 0 or 1, the conclusion of polarization follows. The key point is that the martingale $I(W_n)$ is a random walk in $[0,1]$ and it is `unstable at any point $I(W) \in (0,1)$ as it must move at least $I(W^+)-I(W)>0$ in this range. The plot of $I(W^+)-I(W)>0$ for different values of $I(W)$ is provided in Figure \ref{gap_ari}. 

Thus, the larger the spread $I(W^+)-I(W)$, the more unstable the martingale is at non-extremal points and the faster it should converge to the extremes (i.e., polarized channels).  To see why this is connected to the object of study of this paper, we need one more aspect about polar codes. 

\begin{figure}
\begin{center}
\includegraphics[scale=1.7]{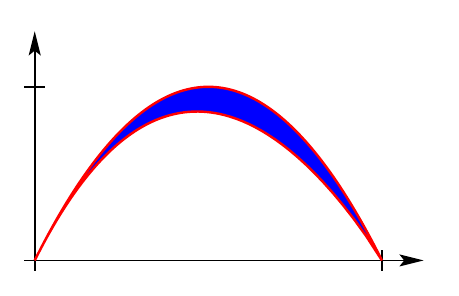}
\caption{Plot of $I(W)$ (horizontal axis) vs.\ $I(W^+)- I(W)$ for all possible binary input channels (the tick on the horizontal axis is at 1 and the tick on vertical axis is at $1/4$).}
\label{gap_ari}
\end{center}
\end{figure}

When considering channels that are `additive noise', polarization can be understood in terms of the noise process rather than the actual channels $W_n$. Consider for example the the binary symmetric channel. When transmitting a codeword $c^n$ on this channel, the output is $Y^n=c^n+ Z^n$, where $Z^n$ has i.i.d.\ Bernoulli components. The polar transform can then be carried over the noise $Z^n$. Since
\begin{align}
&I(U_i; Y^n U^{i-1})  = 1-H((G_n Z^n)_i |  (G_n Z^n)^{i-1}), \label{drop}
\end{align}
the mutual information of the polarized channels is directly obtained from the conditional entropies of the polarized noise vector $G_nZ^n$. The counterpart of this polarization phenomenon is called source polarization \cite{Ari10:isit}. It is extended in \cite{Abb11:ita} to multiple correlated sources.  
For $n=2$, the spread of the two conditional entropies is exactly given by $H(Z+Z') - H(Z)$, where $Z,Z'$ are i.i.d.\ under the noise distribution. 
In Ar{\i}kan and Telatar \cite{AT09:isit}, the rate of convergence of the polar martingale is studied as a function of the block length. Our goal here is to investigate the performance at finite block length, motivated by maximizing the spread at block length $n=1$. When considering non-binary polar codes, that spread is governed by the entropy of a linear combination of i.i.d.\ variables. Preliminary results on this approach were presented in \cite{Abb12:izs}, while the error exponent and scaling law of polar codes have been studied in particular in \cite{Has13:phd} and references therein.

\subsection{Kernels with maximal spread} 

Being interested in the performance of polar codes at finite block length, we start with the optimization of the kernel matrix over $\mF_q$ of block length $n=2$. Namely, we investigate the following optimization problem:
\begin{align}
K^*(W) = \arg\max_{ K \in M_2(\mF_q) } I(W^+(W,K)),  
\end{align}
where $W^+(W,K)$ is the channel $u_2 \mapsto Y_1 Y_2 u_1$, and $(Y_1,Y_2)$ are the output of two independent uses of $W$ when $(x_1,x_2)=(u_1,u_2)K$ are the inputs. We call $K^*$ the 2-optimal kernel for $W$. 

A general kernel is a $2 \times 2$ invertible matrix over $\mF_q$. Let  
$K= 
   \bigl[\begin{smallmatrix} 
      a & b \\
      c & d \\
   \end{smallmatrix}\bigr]$ be such a matrix and let $(U_1, U_2)$ be i.i.d. under $\mu$ over $\mF_q$ and $(X_1, X_2)=(U_1, U_2)K$. Since $K$ is invertible, we have 
\begin{align}
2H(\mu)= H(U_1,U_2)=H(X_1,X_2)=H(X_1)+H(X_2|X_1)
\end{align}
and
\begin{align}
H(X_1)-H(\mu)= H(\mu)-H(X_2|X_1)
\end{align}
which is the entropy spread gained by using the transformation $K$.
To maximize the spread, one may maximize $H(X_1)= H(aU_1 + cU_2)$ over the choice of $a$ and $c$, or simply $H(U_1 + c U_2)$ over the choice of $c$. Hence, the maximization problem depends only on the variable $c$, ($a$ can be set to 1, and $b, d$ only need to ensure that $K$ is invertible), which
leads to a kernel of the form
$K= 
   \bigl[\begin{smallmatrix} 
      1 & 0 \\
     c & 1 \\
   \end{smallmatrix}\bigr]$.
Note that to maximize the spread, one may alternatively minimize $H(X_2|X_1)=H(U_2|U_1 + c U_2)$. 

We consider in particular channels which are `additive noise', in which case one can equivalently study the `source' version of this problem as follows:
\begin{align}
\lambda^*(\mu) = \arg\max_{ \lambda \in \mF_q } H(U_1+ \lambda U_2),   
\end{align}
where $U_1, U_2$ are i.i.d.\ under $\mu$. 
As discussed above, this is related with the previous problem by choosing $$K^*(W)= \begin{bmatrix} 
      1 & 0 \\
     \lambda^*(\mu) & 1 \\
   \end{bmatrix},$$ where $\mu$ is the distribution of the noise of the channel $W$. 


Our first observation about the optimal coefficients $\lambda^*(\mu)$ is in the context of 
$\mF_3$, and follows immediately from Theorem~\ref{thm:Z3}.

\begin{corol}
For a probability distribution $\mu$ over $\mF_3$, 
$$\lambda^*(\mu) = 2$$
if $\mu(1)\neq \mu(2)$, and $\lambda^*(\mu) = \{1,2\}$ if $\mu(1) = \mu(2)$.
\end{corol}

 
Figure \ref{imp3} illustrates the improvements of the error probability of a polar code using the kernel $\bigl[\begin{smallmatrix} 
      1 & 0 \\
      2 & 1 \\
   \end{smallmatrix}\bigr]$ instead of 
   $\bigl[\begin{smallmatrix} 
      1 & 0 \\
      1 & 1 \\
   \end{smallmatrix}\bigr]$ 
   for a block length $n=1024$ when the channel is an additive noise channel over $\mF_3$ with noise distribution $\{0.7, 0.3, 0\}$.
\begin{figure}
\begin{center}
\includegraphics[scale=.55]{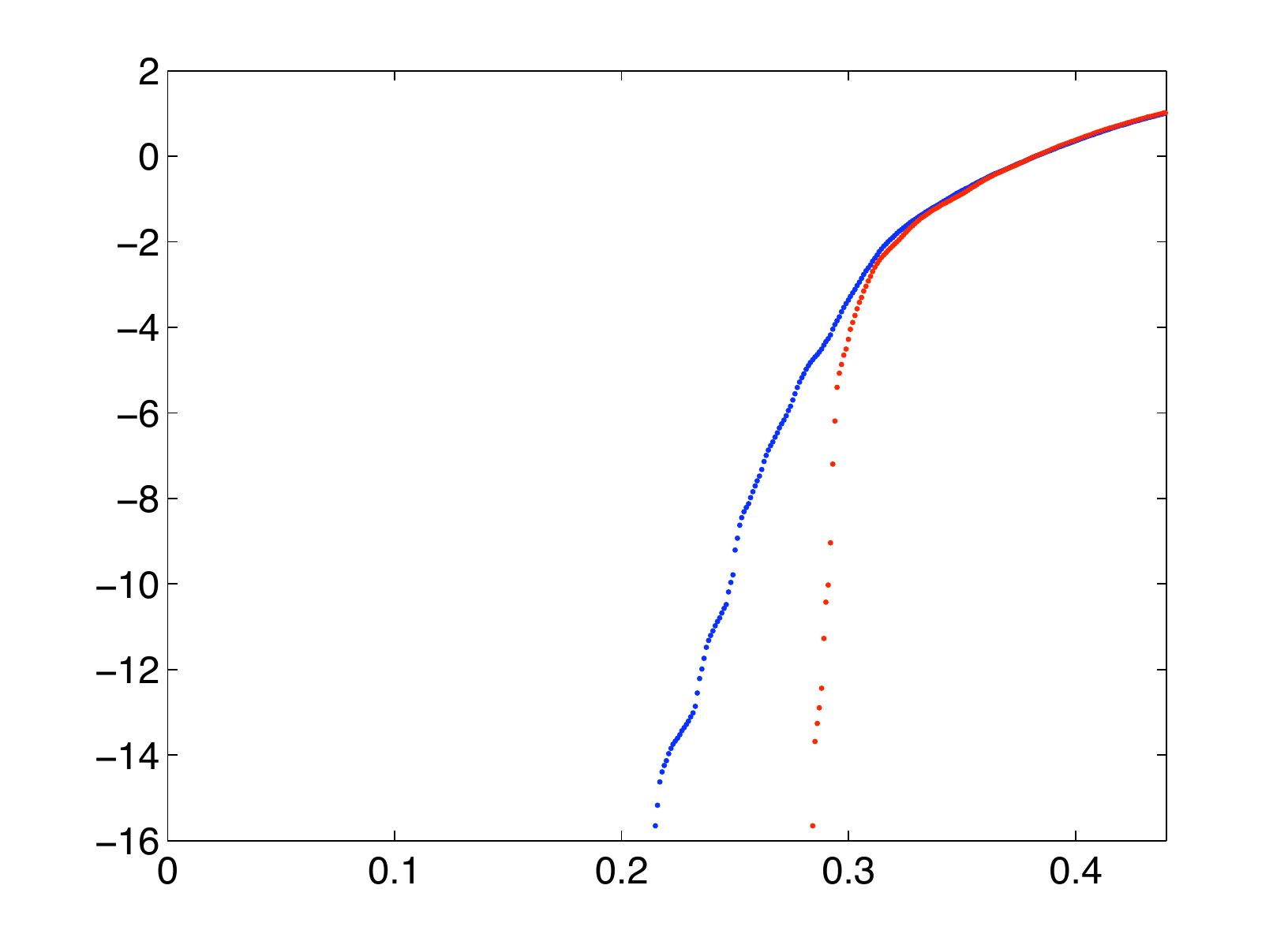}
\caption{For an additive noise channel over $\mF_3$ with noise distribution $\{0.7, 0.3, 0\}$, 
the block error probability (in $\log_{10}$ scale) of a polar code with block length of $n=1024$
is plotted against the rate of the code. The red curve (lower curve) is for the polar code 
using the 2-optimal kernel  whereas the blue curve is for the polar code using the original kernel.
}
\label{imp3}
\end{center}
\end{figure}

When $\mu$ is over $\mF_q$ with $q \geq 5$, $\lambda^*(\mu)$ varies with $\mu$.
For example, one can check numerically that for the distribution $\{0.8,0.1,0.1,0,0\}$ we have $\lambda^*=4$, whereas for 
the distribution $\{0.7,0.2,0.1,0,0\}$ we have $\lambda^*=\{2,3\}$.
Thus finding a solution to the problem of determining $\lambda^*(\mu)$
for general probability distributions $\mu$ on $\mF_q$ seems not so easy.
Nonetheless, for a certain class of probability distributions $\mu$, we can 
identify  $\lambda^*(\mu)$ explicitly using the following observation.


\begin{prop}\label{support}
Let $\mu$ be a probability distribution over $\mF_q$ with support $S_{\mu}$. If there exists $\gamma \in \mF_q$ such that 
\begin{align}
|S_{\mu} + \gamma S_{\mu}| = |S_{\mu}|^2 \label{cond}
\end{align}
then
\begin{align}
H(U_2| U_1 + \gamma U_2)  = 0
\end{align}
where $U_1, U_2$ are i.i.d. under $\mu$.
\end{prop}

\begin{proof}
The condition $|S_{\mu} + \gamma S_{\mu}| = |S_{\mu}|^2$ ensures that knowing $u_1 + \gamma u_2$ with $u_1,u_2 \in S_{\mu}$ allows to exactly recover both $u_1$ and $u_2$.
\end{proof}

\begin{remark}
The condition on the support could be simplified but as such it makes the conclusion of Proposition~\ref{support} immediate. 
Also note that $\gamma$ such that $H(U_2| U_1 + \gamma U_2)  = 0$ is clearly optimal to maximize the spread, 
i.e., it maximizes $H(U_1 + \gamma U_2)$.  
\end{remark}

\begin{figure}
\begin{center}
\includegraphics[scale=.55]{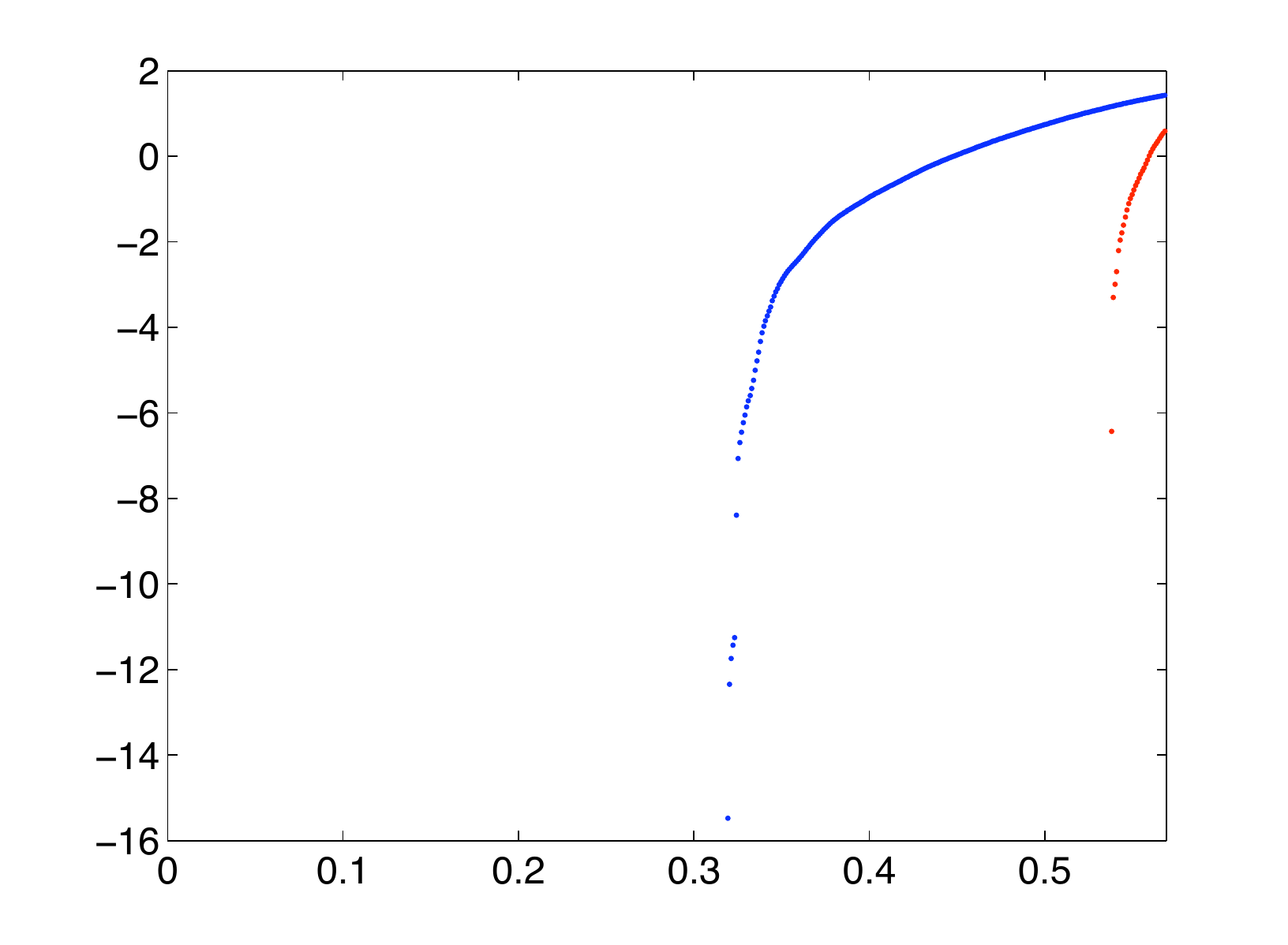}
\caption{For an additive noise channel over $\mF_5$ with noise distribution $\{0.5, 0.5, 0, 0, 0\}$
(i.e., taking any symbol of $\mF_5$ to itself with probability $\half$ and shifting any symbol circularly with probability $\half$), 
the block error probability (in $\log_{10}$ scale) of a polar code with block length of $n=1024$
is plotted against the rate of the code. The red curve (lower curve) is for the polar code 
using the 2-optimal kernel  whereas the blue curve is for the polar code using the original kernel.}
\label{gap5}
\end{center}
\end{figure}

Let us consider some examples of distributions satisfying \eqref{cond}:
\begin{enumerate}
\item Let $\mu$ over $\mF_5$ be such that $S_{\mu}=\{0,1\}$. Picking $\gamma = 2$, one obtains $2 S_{\mu}=\{0,2\}$ and $S_{\mu}+ 2 S_{\mu}=\{0,1,2,3\}$, and \eqref{cond} is verified. In this case, using $\gamma=1$ can only provide a strictly smaller spread since it will not set $H(U_2| U_1 + \gamma U_2)  = 0$. It is hence better to use the 2-optimal kernel $\bigl[\begin{smallmatrix} 
      1 & 0 \\
      2 & 1 \\
   \end{smallmatrix}\bigr]$ rather than the original kernel $\bigl[\begin{smallmatrix} 
      1 & 0 \\
      1 & 1 \\
   \end{smallmatrix}\bigr]$. As illustrated in Figure \ref{gap5}, this leads to significant improvements in the error probability at finite block length. Also note that a channel with noise $\mu$ satisfying \eqref{cond} has positive zero-error capacity, which is captured by the 2-optimal kernel as shown with the rapid drop of the error probability (it is $0$ at low enough rates since half of the synthesized channels have noise entropy exactly zero). 
If $\mu$ is close to a distribution satisfying \eqref{cond}, the error probability can also be significantly improved with respect to the original kernel $\bigl[\begin{smallmatrix} 
      1 & 0 \\
      1 & 1 \\
   \end{smallmatrix}\bigr]$. 
   
\item Over $\mF_{11}$, let $\mu$ be such that $S_{\mu}=\{0,1,2\}$. Picking $\gamma = 2$, one obtains $2 S_{\mu}=\{0,2,4\}$ and \eqref{cond} does not hold. However, picking $\gamma = 3$ leads to $3 S_{\mu}=\{0,3,6\}$ and \eqref{cond} holds. Therefore, the choice of $\gamma$ varies with respect to $q$.

\item Over general $\mF_q$, let $k=\lfloor \sqrt{q-1} \rfloor$. If $S_{\mu}=\{0,1,\dots, k-1\}$, we can see that $\gamma = k$ will satisfy \eqref{cond}.   
\end{enumerate}


In conclusion, we have shown that over $\mF_q$, the martingale spread can be significantly 
enlarged by using 2-optimal kernels rather than the original kernel $\bigl[\begin{smallmatrix} 
      1 & 0 \\
      1 & 1 \\
   \end{smallmatrix}\bigr]$. 
Moreover, we have observed that this can lead to significant improvements on the error probability of polar codes, even at low 
block length ($n=1024$). For additive noise channels, while the improvement is significant when the noise distribution is 
concentrated on ``small'' support, the improvement may not be as significant for distributions that are more more spread out.
   

\section*{Acknowledgement}
The authors would like to thank Emre Telatar for helpful comments and stimulating discussions.


\begin{thebibliography}{10}

\bibitem{Abb11:ita}
E.~Abbe.
\newblock Randomness and dependencies extraction via polarization.
\newblock In {\em Proc.\ Information Theory and Applications Workshop (ITA)},
  pages 1--7, 2011.

\bibitem{Abb12:izs}
E.~Abbe.
\newblock Polar martingale of maximal spread.
\newblock In {\em International Zurich Seminar}, 2012.

\bibitem{Ari09}
E.~Ar{\i}kan.
\newblock Channel polarization: a method for constructing capacity-achieving
  codes for symmetric binary-input memoryless channels.
\newblock {\em IEEE Trans. Inform. Theory}, 55(7):3051--3073, 2009.

\bibitem{Ari10:isit}
E.~Ar{\i}kan.
\newblock Source polarization.
\newblock In {\em Proc. IEEE Intl. Symp. Inform. Theory}, pages 899--903.
  Austin, Texas, June 2010.

\bibitem{AT09:isit}
E.~Ar{\i}kan and E.~Telatar.
\newblock On the rate of channel polarization.
\newblock In {\em Proc. IEEE Intl. Symp. Inform. Theory}, pages 1493--1495.
  Seoul, Korea, 2009.

\bibitem{CZ08}
A.~S. Cohen and R.~Zamir.
\newblock Entropy amplification property and the loss for writing on dirty
  paper.
\newblock {\em IEEE Trans. Inform. Theory}, 54(4):1477--1487, 2008.

\bibitem{CT06:book}
T.~M. Cover and J.~A. Thomas.
\newblock {\em Elements of information theory}.
\newblock Wiley-Interscience [John Wiley \& Sons], Hoboken, NJ, second edition,
  2006.

\bibitem{EO09}
R.~H. Etkin and E.~Ordentlich.
\newblock The degrees-of-freedom of the {$K$}-user {G}aussian interference
  channel is discontinuous at rational channel coefficients.
\newblock {\em IEEE Trans. Inform. Theory}, 55(11):4932--4946, 2009.

\bibitem{HAT14}
S.~Haghighatshoar, E.~Abbe, and E.~Telatar.
\newblock A new entropy power inequality for integer-valued random variables.
\newblock {\em IEEE Trans. Inform. Th.}, 60(7):3787--3796, July 2014.

\bibitem{Has13:phd}
H.~Hassani.
\newblock {\em Polarization and Spatial Coupling: Two Techniques to Boost
  Performance}.
\newblock PhD thesis, EPFL, n.\ 5706, 2013.

\bibitem{HM09}
P.~Hegarty and S.~J. Miller.
\newblock When almost all sets are difference dominated.
\newblock {\em Random Structures Algorithms}, 35(1):118--136, 2009.

\bibitem{Heg07}
P.~V. Hegarty.
\newblock Some explicit constructions of sets with more sums than differences.
\newblock {\em Acta Arith.}, 130(1):61--77, 2007.

\bibitem{JA14}
V.~Jog and V.~Anantharam.
\newblock The entropy power inequality and {M}rs. {G}erber's lemma for groups
  of order {$2^n$}.
\newblock {\em IEEE Trans. Inform. Theory}, 60(7):3773--3786, 2014.

\bibitem{KM14}
I.~Kontoyiannis and M.~Madiman.
\newblock Sumset and inverse sumset inequalities for differential entropy and
  mutual information.
\newblock {\em IEEE Trans. Inform. Theory}, 60(8):4503--4514, August 2014.

\bibitem{LP08:ieeei}
A.~Lapidoth and G.~Pete.
\newblock On the entropy of the sum and of the difference of two independent
  random variables.
\newblock {\em Proc. IEEEI 2008, Eilat, Israel}, 2008.

\bibitem{Mad08:itw}
M.~Madiman.
\newblock On the entropy of sums.
\newblock In {\em Proc. IEEE Inform. Theory Workshop}, pages 303--307. Porto,
  Portugal, 2008.

\bibitem{MK10:isit}
M.~Madiman and I.~Kontoyiannis.
\newblock The entropies of the sum and the difference of two {IID} random
  variables are not too different.
\newblock In {\em Proc. IEEE Intl. Symp. Inform. Theory}, Austin, Texas, June
  2010.

\bibitem{MK15}
M.~Madiman and I.~Kontoyiannis.
\newblock {Entropy bounds on abelian groups and the Ruzsa divergence}.
\newblock {\em Preprint, {\tt arXiv:1508.04089}}, 2015.

\bibitem{MMT12}
M.~Madiman, A.~Marcus, and P.~Tetali.
\newblock Entropy and set cardinality inequalities for partition-determined
  functions.
\newblock {\em Random Struct. Alg.}, 40:399--424, 2012.

\bibitem{MT10}
M.~Madiman and P.~Tetali.
\newblock Information inequalities for joint distributions, with
  interpretations and applications.
\newblock {\em IEEE Trans. Inform. Theory}, 56(6):2699--2713, June 2010.

\bibitem{Mar69}
J.~Marica.
\newblock On a conjecture of {C}onway.
\newblock {\em Canad. Math. Bull.}, 12:233--234, 1969.

\bibitem{MO07}
G.~Martin and K.~O'Bryant.
\newblock Many sets have more sums than differences.
\newblock In {\em Additive combinatorics}, volume~43 of {\em CRM Proc. Lecture
  Notes}, pages 287--305. Amer. Math. Soc., Providence, RI, 2007.

\bibitem{MT10:itw}
R.~Mori and T.~Tanaka.
\newblock {Non-Binary Polar Codes using Reed-Solomon Codes and Algebraic
  Geometry Codes}.
\newblock In {\em Proc. 2010 IEEE Inform. Theory Workshop}. Dublin, 2010.

\bibitem{MOS10}
E.~Mossel.
\newblock Gaussian bounds for noise correlation of functions.
\newblock {\em Geom. Funct. Anal.}, 19(6):1713--1756, 2010.

\bibitem{Nat07:2}
M.~B. Nathanson.
\newblock Problems in additive number theory. {I}.
\newblock In {\em Additive combinatorics}, volume~43 of {\em CRM Proc. Lecture
  Notes}, pages 263--270. Amer. Math. Soc., Providence, RI, 2007.

\bibitem{Nat07:1}
M.~B. Nathanson.
\newblock Sets with more sums than differences.
\newblock {\em Integers}, 7:A5, 24, 2007.

\bibitem{PB13}
W.~Park and A.~Barg.
\newblock Polar codes for {$q$}-ary channels, {$q=2^r$}.
\newblock {\em IEEE Trans. Inform. Theory}, 59(2):955--969, 2013.

\bibitem{Pic42}
S.~Piccard.
\newblock {\em Sur des ensembles parfaits}.
\newblock M\'em. Univ. Neuch\^atel, vol. 16. Secr\'etariat de l'Universit\'e,
  Neuch\^atel, 1942.

\bibitem{FP73}
V.~P. Pigarev and G.~A. Fre{\u\i}man.
\newblock The relation between the invariants {$R$} and {$T$}.
\newblock In {\em Number-theoretic studies in the {M}arkov spectrum and in the
  structural theory of set addition ({R}ussian)}, pages 172--174. Kalinin. Gos.
  Univ., Moscow, 1973.

\bibitem{Rot06:book}
R.~M. Roth.
\newblock {\em Introduction to Coding Theory}.
\newblock Cambridge University Press, 2006.

\bibitem{Ruz78}
I.~Z. Ruzsa.
\newblock On the cardinality of {$A+A$} and {$A-A$}.
\newblock In {\em Combinatorics ({P}roc. {F}ifth {H}ungarian {C}olloq.,
  {K}eszthely, 1976), {V}ol. {II}}, volume~18 of {\em Colloq. Math. Soc.
  J\'anos Bolyai}, pages 933--938. North-Holland, Amsterdam-New York, 1978.

\bibitem{Ruz92:1}
I.~Z. Ruzsa.
\newblock On the number of sums and differences.
\newblock {\em Acta Math. Hungar.}, 59(3-4):439--447, 1992.

\bibitem{Ruz09}
I.~Z. Ruzsa.
\newblock Entropy and sumsets.
\newblock {\em Random Struct. Alg.}, 34:1--10, 2009.

\bibitem{SP13}
A.~G. Sahebi and S.~S. Pradhan.
\newblock Multilevel channel polarization for arbitrary discrete memoryless
  channels.
\newblock {\em IEEE Trans. Inform. Theory}, 59(12):7839--7857, 2013.

\bibitem{SS03:1:book}
E.~M. Stein and R.~Shakarchi.
\newblock {\em Fourier analysis}, volume~1 of {\em Princeton Lectures in
  Analysis}.
\newblock Princeton University Press, Princeton, NJ, 2003.
\newblock An introduction.

\bibitem{Ste73}
S.~K. Stein.
\newblock The cardinalities of {$A+A$} and {$A-A$}.
\newblock {\em Canad. Math. Bull.}, 16:343--345, 1973.

\bibitem{Tao10}
T.~Tao.
\newblock Sumset and inverse sumset theory for {S}hannon entropy.
\newblock {\em Combin. Probab. Comput.}, 19(4):603--639, 2010.

\bibitem{TV06:book}
T.~Tao and V.~Vu.
\newblock {\em Additive combinatorics}, volume 105 of {\em Cambridge Studies in
  Advanced Mathematics}.
\newblock Cambridge University Press, Cambridge, 2006.

\bibitem{WWM14:isit}
L.~Wang, J.~O. Woo, and M.~Madiman.
\newblock A lower bound on the {R{\'e}nyi} entropy of convolutions in the
  integers.
\newblock In {\em Proc. IEEE Intl. Symp. Inform. Theory}, pages 2829--2833.
  Honolulu, Hawaii, July 2014.

\bibitem{WM15:isit}
J.~O. Woo and M.~Madiman.
\newblock A discrete entropy power inequality for uniform distributions.
\newblock In {\em Proc. IEEE Intl. Symp. Inform. Theory}, Hong Kong, China,
  June 2015.

\bibitem{WSV12}
Y.~Wu, S.~Shamai, and S.~Verd\'u.
\newblock A general formula for the degrees of freedom of the interference
  channel.
\newblock {\em Preprint}, 2012.

\bibitem{Zha10:3}
Y.~Zhao.
\newblock Constructing {MSTD} sets using bidirectional ballot sequences.
\newblock {\em J. Number Theory}, 130(5):1212--1220, 2010.

\bibitem{Zha10:2}
Y.~Zhao.
\newblock Counting {MSTD} sets in finite abelian groups.
\newblock {\em J. Number Theory}, 130(10):2308--2322, 2010.

\end{thebibliography}

\end{document}